\documentclass[submission]{eptcs}
\providecommand{\event}{ArXiv Pre-print} % Name of the event you are submitting to
\usepackage{breakurl}             % Not needed if you use pdflatex only.

\usepackage{amsfonts, amssymb, amscd, amsthm, amsmath, xspace}
\usepackage{IEEEtrantools}

\usepackage{algorithm}
\usepackage{algorithmic}

\usepackage{color}
\usepackage{graphicx}
\usepackage{epstopdf}

\newtheorem{theorem}{Theorem}
\newtheorem{lemma}{Lemma}

\newtheorem{remark}{Remark}

\newtheorem{definition}{Definition}
\newtheorem{proposition}{Proposition}
\newtheorem{fact}{Fact}

\newcommand{\etal}{{\it et al.}}
\newcommand{\ie}{{\it i.e.}}
\newcommand{\eg}{{\it e.g.}}

\newcommand{\GF}[1]{\mathbb{F}_{#1}} % Galois field

\newcommand{\code}{\mathcal{C}} % Code C
\newcommand{\codepts}{C} % Code points C
\newcommand{\codeptsj}[1]{C_{#1}} % Subset of code points C_j
\newcommand{\codeptsinfoj}{I_{j}} % Subset of info points I_j

\newcommand{\codep}{\mathcal{C}^{'}} % Code C'
\newcommand{\codeptsp}{C^{'}} % Code points C'

\newcommand{\vect}[1]{\mathbf{#1}} %vector 
\newcommand{\cw}{\mathbf{c}} %vector c
\newcommand{\supp}[1]{\textsf{Supp}\left(#1\right)} % support
\newcommand{\wt}[1]{\textsf{wt}\left(#1\right)} % wt
\newcommand{\dims}[1]{\textsf{dim}\left(#1\right)} % dimension
\newcommand{\loc}[1]{\textsf{Loc}\left(#1\right)} % locality
\newcommand{\rnk}[1]{\textsf{rank}\left(#1\right)} % rank

\newcommand{\rep}[2]{{R}_{#1}\left(#2\right)} % repair group
\newcommand{\gam}[1]{\Gamma\left(#1\right)} % repair group + c_i
\newcommand{\subspace}[1]{\langle #1 \rangle} % < . >
\newcommand{\nkd}{(n,k,d)} % (n,k,d)
\newcommand{\nkdq}{[n,k,d]_q} % [n,k,d]_q
\newcommand{\ri}{r_i} % r_i
\newcommand{\cwi}{\vect{c}_i} % c_i
\newcommand{\cwl}{\vect{c}_l} % c_l
\newcommand{\cwset}[1]{\vect{c}_{#1}} % c_{T}

\newcommand{\rprof}{\vect{r}} % locality profile r
\newcommand{\kprof}{\vect{k}} % locality profile k
\newcommand{\kprofreq}{\tilde{\vect{k}}} % locality requirement \tilde{k}
\newcommand{\kprofopt}{\kprof^{*}} % optimal loc prof k*
\newcommand{\kprofp}{\kprof^{'}} % optimal loc prof k'
\newcommand{\kprofpp}{\kprof^{''}} % optimal 'good' loc prof k''

\newcommand{\ceillr}[1]{\left\lceil#1\right\rceil} % left ceil, right ceil
\newcommand{\floorlr}[1]{\left\lfloor#1\right\rfloor} % left ceil, right ceil
\newcommand{\setS}{S} % subset S
\newcommand{\setT}{T} % subset T
\newcommand{\setSi}[1]{\setS_{#1}} % subset S_i
\newcommand{\setTi}[1]{\setT_{#1}} % subset T_i
\newcommand{\si}{s_i} % s_i, incremental size
\newcommand{\ti}{t_i} % t_i, incremental rank

\newcommand{\lj}[1]{l_{#1}} % l_j
\newcommand{\kj}[1]{k_{#1}} % k_j
\newcommand{\ktj}[1]{\tilde{k}_{#1}} % \tilde{k}_j
\newcommand{\ksj}[1]{{k}^{*}_{#1}} % k*_j
\newcommand{\kpj}[1]{{k}^{'}_{#1}} % k'_j
\newcommand{\kppj}[1]{{k}^{''}_{#1}} % k''_j

\newcommand{\bj}[1]{\beta_{#1}} % beta_j
\newcommand{\gj}[1]{\gamma_{#1}} % gamma_j
\newcommand{\jp}[1]{j_{#1}} % j_p
\newcommand{\ej}[1]{\vect{e}_{#1}} % e_j
\newcommand{\pj}[1]{\vect{p}_{#1}} % p_j
\newcommand{\xS}[2]{#1|_{#2}} % x |_{s}

\newcommand{\ra}{r_a} % r_a
\newcommand{\nj}[1]{n_{#1}} % n_j
\newcommand{\nprof}{\vect{n}} % locality profile n
\newcommand{\nhj}[1]{\hat{n}_{#1}} % n_hat_j
\newcommand{\nhpj}[1]{\hat{n}'_{#1}} % n'_hat_j
\newcommand{\khj}[1]{\hat{k}_{#1}} % k_hat_j
\newcommand{\khpj}[1]{\hat{k}'_{#1}} % k'_hat_j
\newcommand{\setSp}{S^{'}} % subset S'

\newcommand{\GFm}[1]{\mathbb{F}_{{#1}^m}} % Galois field extension q^m
\newcommand{\GFnm}[1]{\mathbb{F}_{#1}^{N\times m}} % Galois field extension q^{N x m}

\newcommand{\dr}[2]{d_R\left(#1,#2\right)} % rank distance
\newcommand{\m}{\vect{m}} % message m
\newcommand{\Nj}[1]{N_{#1}} % N_j
\newcommand{\Npj}[1]{N^{'}_{#1}} % N'_j
\newcommand{\cwG}{\mathbf{c}_{\textrm{Gab}}} %vector c_{Gab}
\newcommand{\cwGj}[1]{\mathbf{c}_{\textrm{Gab}}^{#1}} %vector c_{Gab}^j

\title{Codes with Unequal Locality}
\author{Swanand Kadhe and Alex Sprintson
\institute{Electrical and Computer Engineering\\
Texas A\&M University}
\email{kswanand1@tamu.edu, spalex@tamu.edu}
}
\def\titlerunning{Codes with Unequal Locality}
\def\authorrunning{S. Kadhe and A. Sprintson}

\begin{document}
\maketitle

\begin{abstract}
For a code $\code$, its $i$-th symbol is said to have locality $r$ if its value can be recovered by accessing some other $r$ symbols of $\code$. Locally repairable codes (LRCs) are the family of codes such that every symbol has locality $r$.  

In this paper, we focus on (linear) codes whose individual symbols can be partitioned into subsets such that symbols in one subset have different locality than the ones in other. We call such codes as \emph{codes with unequal locality}. For codes with \emph{unequal information locality}, we compute a tight upper bound on the minimum distance as a function of number of information symbols of each locality. We demonstrate that the construction of Pyramid codes can be adapted to design codes with unequal information locality that achieve the minimum distance bound. This result generalizes the classical result of Gopalan \etal~for codes with unequal locality.
Next, we consider codes with \emph{unequal all symbol locality}, and establish an upper bound on the minimum distance as a function of number of symbols of each locality. We show that the construction based on rank-metric codes by Silberstein \etal~can be adapted to obtain codes with unequal all symbol locality that achieve the minimum distance bound. Finally, we introduce the concept of \emph{locality requirement} on a code, which can be viewed as a recoverability requirement on symbols. Information locality requirement on a code essentially specifies the minimum number of information symbols of different localities that must be present in the code. We present a greedy algorithm that assigns localities to information symbols so as to maximize the minimum distance among all codes that satisfy a given locality requirement.
\end{abstract}

\section{Introduction}
\label{sec:intro}
Coding for distributed storage has recently attracted significant research attention with a focus on the problem of recovery from storage node failures. The thrust has been on characterizing fundamental limits and designing associated coding schemes for one or more of the following metrics that are crucial in the node repair process: (a) {\it repair bandwidth} -- the amount of data downloaded during failed node repair~\cite{Dimakis:10, Dimakis:11}; (b) disk I/O -- the number of bits read from the nodes participating in the repair process~\cite{Tamo:Zigzag13,Khan:11}; and (c) {\it repair locality} -- the number of nodes participating in the repair process~\cite{Gopalan:12,Oggier:11}. 

In this paper, we focus on the metric of repair locality and a class of codes designed in the context of this metric, known as {\it locally repairable codes} (LRCs). Consider a block code of length $n$ that encodes $k$ information symbols. A symbol $i$ is said to have {\it locality} $\ri$ if it can be recovered by accessing $\ri$ other symbols in the code. We say that a code has {\it information locality} $r$ if each of its $k$ information symbols has locality at most $r$. Similarly, we say that a code has {\it all-symbol locality} $r$ if each of its $n$ symbols has locality at most $r$. 

Codes with small locality were introduced in~\cite{Huang:07, Han:07} (see also~\cite{Oggier:11}). The study of the locality property was galvanized with the pioneering work of Gopalan \etal~\cite{Gopalan:12}. One of their key contributions was to establish a trade-off between the minimum distance of a code and its information locality analogous to the classical Singleton bound. In particular, the authors showed that for a (scalar) linear $(n,k)$ code having information locality $r$, the minimum distance $d$ of the code is upper bounded as
\begin{equation}
\label{eq:Gopalan}
d \leq n - k - \ceillr{\frac{k}{r}} + 2.
\end{equation}
They also demonstrated that the Pyramid code construction in~\cite{Huang:07} achieves this bound.
Since then, a series of results have extended the code distance bound for a given locality for various types of codes along with corresponding {\it optimal} code constructions achieving the distance bound. We give a brief (far from complete) overview of some of these results below.

{\bf Related work:} 
The distance bound was generalized for codes with multiple local parities in~\cite{Prakash:12}, universal (scalar/vector linear, nonlinear) codes in~\cite{Papailiopoulos:14}, universal codes with multiple parities in~\cite{Rawat:14,Kamath:14}. An integer programming based bound was established in~\cite{Wang:15}. Almost all of these works also presented optimal code constructions. Furthermore, a large number of other optimal code constructions have been presented, see \eg,~\cite{Silberstein:13,Tamo:LRC13,Ernvall:14,TamoB:14,Goparaju:14,Song:14,Kuijper:14,Huang:15,ZehY:15,SilbersteinZ:15}. The notion of locality was extended to multiple recovery groups (also known as, availability) in~\cite{Rawat:Availability14,Tamo:Availability14}, and for the case of multiple failures, to sequential repair in~\cite{Prakash:14} and hierarchical repair in~\cite{Sasidharan:15}. The Singleton-like bound was extended to accommodate the alphabet size in~\cite{CadambeM:15}.

{\bf Our contributions:} 
In previous works, the locality of a code is characterized by a single parameter $r$. Inspired from the notion of {\it unequal error protection}, we are interested in investigating linear codes, in which, different subsets of symbols possess different localities. We refer to such codes as {\it codes with unequal locality}. For example, consider a $(15,11)$ code whose 4 information symbols have locality 2, 3 information symbols have locality 3, and 4 information symbols have locality 4 (with no constraint on the locality of parity symbols). Under the classical terminology, such a code would be characterized as a code with information locality 4. However, it is not clear if the distance bound given in~\eqref{eq:Gopalan} is tight for the case of unequal localities. Our main goal is to compute a tight upper bound on the minimum distance of such codes with unequal locality. 

Codes with unequal locality are practically appealing in scenarios when important information symbols, \eg, symbols of {\it hot data}, need to be repaired quickly; whereas, recovering less important symbols can involve more overhead. Moreover, these types of codes can be useful in reducing download latency for hot data. For instance, references~\cite{Joshi:14,KSS:15} study storage codes from queueing theoretic perspective to analyze download latency. 

Our key contributions are summarized as follows.  To characterize a code with unequal information locality, we define a notion of {\it information locality profile} of a code. We say that a code has an information locality profile \mbox{$\kprof = \{\kj{1},\ldots,\kj{r}\}$} if it contains $\kj{j}$ information symbols of locality $j$ for $1\leq j\leq r$. For example, a code having 5 information symbols of locality 2, and 6 information symbols of locality 4 would have an information locality profile $\{0,5,0,6\}$. 
For scalar linear codes, we establish an upper bound on the minimum distance as a function of  information locality profile $\kprof = \{\kj{1},\ldots,\kj{r}\}$ as follows (Theorem~\ref{thm:dist-bound-info-loc}).
\begin{equation}
\label{eq:bound-1}
d \leq n - k - \sum_{j=1}^{r}\ceillr{\frac{\kj{j}}{j}} + 2.
\end{equation} 
We demonstrate that the Pyramid code construction in~\cite{Huang:07} can be adapted to design unequal locality codes that are distance-wise optimal according to the bound above. 

When parity symbols also have locality constraints, we can analogously define an {\it all symbol locality profile} of a code. W say that a code has an all symbol locality profile $\nprof = \{\nj{1},\ldots,\nj{r}\}$ if it contains $\kj{j}$ information symbols of locality $j$ for $1\leq j\leq r$. For instance, consider a $(15,11)$ code that has 6 symbols of locality 2, 4 symbols of locality 3, and 5 symbols of locality 4. Its all symbol locality profile would be $\{0,6,4,5\}$. We compute an upper bound on the minimum distance for scalar linear codes as a function of all symbol locality profile, which has the following form (Theorem~\ref{thm:all-symbol-loc-bound}).\footnote{In an parallel and independent work, Zeh and Yaakobi~\cite{ZehY:16} also consider the problem of computing a bound on minimum distance of codes with unequal all symbol locality, referred in their work as multiple locality codes. Their bound~\cite[Theorem 8]{ZehY:16} has a similar form as we get. In addition,~\cite{ZehY:16} extends Cadambe-Mazumdar bound in~\cite{CadambeM:15} for codes with multiple localities, and present several optimal code constructions.}
\begin{equation}
\label{eq:bound-2}
d \leq n - k + 2 - \sum_{j=1}^{r-1}\ceillr{\frac{\nj{j}}{j+1}} - \ceillr{\frac{k - \sum_{i=1}^{r-1}\left(\nj{j} - \ceillr{\frac{\nj{j}}{j+1}}\right)}{r}}.
\end{equation}
We adapt the construction in~\cite{Silberstein:13}, which uses a maximum rank distance (MRD) code as an outer code and a maximum distance separable (MDS) code as inner code, to construct codes with unequal all symbol locality that are optimal with respect to the above bound.

Finally, we introduce a concept of {\it information locality requirement}. To motivate this, consider a scenario where we need to design a linear code of dimension $k=11$ such that $\ktj{3} = 5$ information symbols must have locality at most $3$, and the remaining $\ktj{4} = 6$ information symbols must have locality at most $4$. Collectively, we can specify this as a locality requirement of $\kprofreq = \{0,0,5,6\}$. Notice that this is equivalent to a requirement as a code must contain at least 5 symbols of locality up to 3, and at least 11 symbols of locality up to 4. In general, a locality requirement of $\kprofreq = \{\ktj{1},\ldots,\ktj{r}\}$ means that a code should contain at least $\sum_{i=1}^{i}\ktj{j}$ symbols of locality up to $i$ for each $1\leq i\leq r$, or, in other words, $\ktj{j}$ information symbols should have locality at most $j$. 

One can design codes with various information locality profiles that would satisfy this requirement. For examples, the locality requirement of $\kprofreq = \{0,0,5,6\}$ is satisfied by locality profiles $\{5,6\}$, $\{0,0,5,6\}$, $\{0,2,9\}$, $\{1,0,6,4\}$, etc. The question what is the maximum value of minimum distance any code with this locality requirement would attain, and can we find an {\it optimal} locality profile which achieves this distance? Note that locality requirement can be viewed as a recoverability requirement for code design. We give a simple greedy algorithm which computes an information locality profile given an information locality requirement.

\section{Preliminaries}
\label{sec:basics}

\subsection{Notation}
\label{sec:notation}
We use the following notation.
\begin{enumerate}
\item For an integer $l$, $[l] = \{1,2,\ldots,l\}$;
\item For a vector $\vect{x}$ and an integer $i$, $\vect{x}(i)$ denotes the $i$-th coordinate of $\vect{x}$, for a matrix $H$ and integers $i, j$, $H(i,j)$ denotes the element in row $i$ and column $j$;
\item For a vector $\vect{x}$, $\supp{\vect{x}}$ denotes its support, \ie, $\supp{\vect{x}} = \{i : \vect{x}(i) \neq 0\}$;
\item For a vector $\vect{x}$, $\wt{\vect{x}}$ denotes its Hamming weight, \ie, $\wt{x} = |\supp{\vect{x}}|$;
\item For vectors $\vect{x}$ and $\vect{y}$, $\vect{x}\cdot\vect{y}$ denotes their dot product;
\item For a set of vectors $\vect{x}_1, \ldots, \vect{x}_m$, $\subspace{\vect{x}_1, \ldots, \vect{x}_m}$ denotes their span, whereas for a matrix $H$, $\subspace{H}$ denotes its row space;
\item For a  vector space $\mathcal{A}$, $\dims{\mathcal{A}}$ denotes its dimension;
\item For a matrix $H$, $\rnk{H}$ denotes the rank of $H$.
\end{enumerate}

\subsection{Codes with Locality}
\label{sec:basics}
Let $\code$ denote a linear $\nkdq$ code over $\GF{q}$ with block-length $n$, dimension $k$, and minimum distance $d$. Let $\cw$ denote a codeword in $\code$. The code can be represented by a set of $n$ (column) vectors $\codepts = \left\{\vect{c}_1,\ldots,\vect{c}_n\right\} \in \GF{q}^k$. The set of vectors must have rank $k$ for $\code$ to have dimension $k$. The $i$-th vector $\cwi$ is referred to as the $i$-th coordinate of $\code$. For any codeword $\cw\in\code$, $\cw(i)$ is said to be the $i$-th symbol of the codeword $\cw$. In the context of locality, we use the terms symbol or coordinate interchangeably. Our main focus is on systematic codes, and we assume that the first $k$ coordinates correspond to the information symbols.

We say that the $i$-th coordinate of a code $\code$ has locality $\ri$ if its value can be recovered from some other $\ri$ coordinates of $\code$. The formal definition of locality is as follows.

\begin{definition}
\label{def:locality}
[Locality] For $\cwi\in\code$, we define $\loc{\cwi}$ to be the smallest integer $\ri$ such that there exists a subset $\rep{}{i}\subset[n]\setminus\{i\}$, $|\rep{}{i}|\leq \ri$, such that $\cwi = \sum_{l\in\rep{}{i}}\lambda_{l}\cwl$, where $\lambda_l \in\GF{q}$ $\forall\: l\in\rep{}{i}$. 
\end{definition}
Note that, if the minimum distance of the code is more than two, then every coordinate has locality at most $k$. 

We say that an $(n,k)$ code has {\it information locality} $r$ if each of its $k$ information symbols has locality at most $r$. Similarly, we say that an $(n,k)$ code has {\it all symbol locality} $r$ if each of its $n$  symbols has locality at most $r$.

\section{Codes with Unequal Information Locality}
\label{sec:info-loc}
In this section, we are interested in systematic codes, whose information symbols can be partitioned into disjoint subsets in such a way that the symbols in one subset have different locality than the symbols in other subset. We say that such codes possess unequal information locality. We can characterize the locality of such codes by listing the locality values of each information symbol. Alternatively, we can consider the list of cardinalities of subset of each locality. We call such a list as the {\it information locality profile} of the code. Formally, the definition is as follows. 

\begin{definition}
\label{def:info-loc-prof}
[Information Locality Profile] Given a systematic $\nkdq$ code $\code$, the information locality profile of $\code$ is defined as a length-$k$ vector $\rprof(\code) = \{r_1,\ldots,r_k\}$, where $r_i$ is the locality of the $i$-th information coordinate of $\code$. Note that $1\leq \ri \leq k$ for each $i\in[k]$, assuming $d\geq2$. 

Alternatively, we can specify the locality profile of $\code$ as a length-$r$ vector $\kprof(\code) = \{\kj{1}, \ldots, \kj{r}\}$, where $r = \max\{r_1, \ldots, r_k\}$ and $\kj{j}$ is the number of information coordinates of locality $j$ for $j\in[r]$. Note that $\forall j\in[r]$, $0\leq k_j\leq k$, $\kj{r} \geq 1$ and $\sum_{j=1}^{r} \kj{j} = k$. 
\end{definition}

\begin{remark}
\label{rem:info-loc-prof}
For a code $\code$ with representation $\codepts$, we can choose any subset of $k$ full-rank coordinates of $\codepts$ to represent information symbols. Without loss of generality, we can always choose the coordinates having smallest overall locality as information coordinates. More specifically, for $1\leq j\leq r$, let $\codeptsj{j}\subset\codepts$ be the subset of coordinates having locality $j$. Set $\codeptsj{0} = \emptyset$. Let 
\begin{equation}
\label{eq:k-j}
\kj{j} =  \rnk{\cup_{i=0}^{j}\codeptsj{i}} - \rnk{\cup_{i=0}^{j-1}\codeptsj{i}}. 
%\dims{\subspace{\cup_{i=0}^{j}\codeptsj{i}}} - \dims{\subspace{\cup_{i=0}^{j-1}\codeptsj{i}}}.
\end{equation}
In other words, $\sum_{i=1}^{j}\kj{i}$ is the rank of the sub-matrix formed by the coordinates having locality up to $j$. Starting with $j=1$, we choose a subset $\codeptsinfoj \subset \codeptsj{j}$ of $\kj{j}$ linearly independent coordinates to represent $\kj{j}$ information symbols, and continue incrementing $j$ till the total rank is $k$. %having locality $j$. 
\end{remark}

\begin{remark}
\label{rem:loc-prof-classical}
%We can alternatively specify the information locality profile of a code as a length-$k$ vector $\rprof = \{r_1, \ldots, r_k\}$, where $r_j$ denotes the locality of the $j$-th information symbol. 
In the classical notion of locality defined by Gopalan \etal~\cite{Gopalan:12}, technically, every symbol can have different locality. However, the (information) locality of a code is parameterized by a single value $r$, which is the largest locality of an (information) symbol. On the other hand, we parameterize the information locality using a length-$k$ vector that specifies the locality of each individual information symbol. We are interested in characterizing a trade-off between the minimum distance of a code and its locality profile vector.
\end{remark}

\subsection{Bound on the Minimum Distance}
\label{sec:dist-bound}
Consider a class of systematic linear codes having an information locality profile $\kprof = \{\kj{1},\ldots,\kj{r}\}$. We are interested in finding an upper bound on the minimum distance as a function of the code length, dimension, and information locality profile. This would be a generalization of the result in~\cite{Gopalan:12} for codes with unequal localities for information symbols.

\begin{theorem}
\label{thm:dist-bound-info-loc}
For any linear code with block-length $n$, dimension $k$, and information locality profile $\kprof = \{k_1,\ldots,k_r\}$, we have
\begin{equation}
\label{eq:dist-bound-info-loc}
d \leq n - k - \sum_{j = 1}^{r} \ceillr{\frac{k_j}{j}} + 2.
\end{equation}
\end{theorem}

\begin{proof}
We build on the proof technique proposed in~\cite{Gopalan:12}. The idea is to construct a large set $\setS\subseteq\codepts$ such that $\rnk{\setS} \leq k - 1$, and then use the following fact.
\begin{fact}
\label{fact}
(\cite{Gopalan:12}) The code $\code$ has minimum distance $d$ if and only if for every $\setS\subseteq\codepts$ such that $\rnk{\setS} \leq k - 1$, we have
\begin{equation}
\label{eq:fact}
|\setS| \leq n - d.
\end{equation}
\end{fact}

Recall that $\rep{}{i}$ denotes a repair group of $\cwi$, and we have $|\rep{}{i}| = \loc{\cwi}$. Define $\gam{i} := \{i \cup \rep{}{i}\}$. Further, for any subset $\setT\subseteq[n]$, define $\cwset{\setT} = \{\cwi\in\codepts : i\in\setT\}$. 

%%%%%%%%%%%%%%%%%%%%%%%%%%%%%%%%%%%%%%%%%%%%%%%%%
\begin{algorithm}[!t]
\caption{Construct set $\setS\subseteq\codepts$ such that $\rnk{\setS} \leq k - 1$}
\label{alg:find-S}
\begin{algorithmic}[1]
\STATE Let $\setSi{0} = \emptyset$, $i = 1$
\WHILE{$\rnk{\setSi{i-1}}\leq k - 2$}
	\STATE{Pick a coordinate $\cwi\in\codepts\setminus\setSi{i}$ having smallest locality} 
	\label{pick-ci}
	\IF{$\rnk{\setSi{i-1} \cup \cwset{\gam{i}}} < k$}	\label{condition}
		\STATE{Set $\setSi{i} = \setSi{i-1} \cup \cwset{\gam{i}}$}
	\ELSE
		\STATE{Pick $\gam{i}' \subset \gam{i}$ such that $\rnk{\setSi{i-1}\cup\cwset{\gam{i}'}} = 				      k -1$}	
		\STATE{Set $\setSi{i} = \setSi{i-1}\cup\cwset{\gam{i}'}$}
	\ENDIF
\STATE{Increment $i$}
\ENDWHILE
\end{algorithmic}
\end{algorithm}
%%%%%%%%%%%%%%%%%%%%%%%%%%%%%%%%%%%%%%%%%%%%%%%%%

We use Algorithm~\ref{alg:find-S} to construct a set $\setS$ such that $\rnk{\setS} < k$. First, note that in line~\ref{pick-ci}, as $\rnk{\setSi{i-1}} \leq k - 2$, and there are $k$ (linearly independent) information symbols, there exists a coordinate $\cwi\notin\setSi{i-1}$. 

Our goal is to find a lower bound on $|\setS|$. Let $l$ be the total number of iterations of Algorithm~\ref{alg:find-S}. Observe that $|\setS| = |\setSi{l}|$. Further, the final set $\setSi{l}$ has $\rnk{\setSi{l}} = k-1$. We define the increment in the size and rank of set $\setSi{i}$ in the $i$-th iteration as follows.
\begin{equation}
\label{eq:si-ti}
\si = |\setSi{i}| - |\setSi{i-1}|, \quad  \ti = \rnk{\setSi{i}} - \rnk{\setSi{i-1}}.
\end{equation}
Note that
\begin{equation}
\label{eq:S-l-total-rank}
|\setSi{l}| = \sum_{i = 1}^{l}\si, \quad \rnk{\setSi{l}} = \sum_{i=1}^{l}\ti = k - 1.
\end{equation}

We consider two cases depending on whether Algorithm~\ref{alg:find-S} reaches the condition in line~\ref{condition}, \ie, $\rnk{\setSi{i-1} \cup \cwset{\gam{i}}} = k$. We note that the condition can be reached only in the last iteration. 

{\bf Case 1:} Suppose we have $\rnk{\setSi{i-1} \cup \cwset{\gam{i}}} \leq k-1$ throughout. Now, in the $i$-th iteration, we add $\cwset{\gam{i}}$ to $\setS$. Thus, $\si \leq \loc{\cwi} + 1$. Further, vectors in $\cwset{\gam{i}}\setminus\setSi{i-1}$ are such that they yield a (possibly zero) vector in $\subspace{\setSi{i-1}}$. Therefore, 
\begin{equation}
\label{eq:ti-bound}
\ti \leq \si - 1 \leq \loc{\cwi}.
\end{equation} 
Using this, we can write
\begin{equation}
\label{eq:S-case-1}
|\setS| = \sum_{i=1}^{l}\si \geq \sum_{i=1}^{l}(\ti+1) = k - 1 + l,
\end{equation}
where the last equality follows from~\eqref{eq:S-l-total-rank}. 

{\it Lower bounding the number of iterations.} Now, to find a lower bound on $|\setS|$, we find a lower bound on $l$. Let $m$ be the locality of the last symbol collected by Algorithm~\ref{alg:find-S}, where $m\in[r]$.
For $1\leq j\leq m$, let $\lj{j}$ be the number of iterations in which Algorithm~\ref{alg:find-S} picks  coordinates of locality $j$. Note that, if $\code$ does not contain any symbol of a particular locality $j$, we set $\lj{j} = 0$. Thus, for each $j$, $0\leq \lj{j}\leq l$, and $l = \sum_{j=1}^{m} \lj{j}$. 

Recall that $\codeptsj{j}\subset\codepts$ is the set of coordinates of locality $j$ (see Remark~\ref{rem:info-loc-prof}). Since the algorithm collects all coordinates of locality up to $j$ before collecting any coordinate of locality $j+1$ for $1\leq j\leq m-1$, we have $\setSi{\sum_{p=1}^{j}l_p} = \cup_{p=1}^j\codeptsj{p}$. Therefore, from~\eqref{eq:k-j}, $\rnk{\setSi{l_1}} = \kj{1}$ and for $2\leq j\leq m-1$, $\rnk{\setSi{\sum_{p=1}^{j}l_p}} - \rnk{\setSi{\sum_{p=1}^{j-1}l_p}} = \kj{j}$. This results in 
\begin{equation}
\label{eq:rank-S-lj}
\rnk{\setSi{\sum_{p=1}^{j}l_p}} = \sum_{p=1}^{j}\kj{p}, \quad \textrm{for}\:\:1\leq j\leq m-1.
\end{equation} 
The above two results can be interpreted as follows. The increment in the rank of $\setS$ by collecting all the coordinates of locality $j$ is $\kj{j}$ for $1\leq j\leq m-1$. The rank of $\setS$, when it contains all the coordinates of locality up to $j$, is $\sum_{p=1}^{j}\kj{p}$.

When the algorithm terminates, it may not have collected all the coordinates of locality $m$. Let $\kj{m}'$ be the increment in the rank of $\setS$ by the coordinates of locality $m$ that are collected by the algorithm. Note that $1\leq \kj{m}' \leq \kj{m}$. 

Note that $\rnk{\setSi{l}} = \rnk{\setSi{\sum_{j=1}^{m-1}\lj{j}}} + \kj{m}'$. Using the fact that $\rnk{\setSi{l}} = k-1$ and~\eqref{eq:rank-S-lj}, we get $k-1 = \sum_{j=1}^{m-1}\kj{j} + \kj{m}'$. On the other hand, by definition of locality profile vector, we have $\sum_{j=1}^{r}\kj{j} = k$. We consider two cases.

Case (1a): $\kj{r} \geq 2 1$. Then, it must be that $m=r$ and $\kj{m}' = \kj{r} - 1$ since $1\leq \kj{m}' \leq \kj{m}$. 

Case (1b): $\kj{r} = 1$. Then, it follows that $m=r-1$, and $\kj{m}' = \kj{r-1}$ since $1\leq \kj{m}' \leq \kj{m}$.

In summary, for $1\leq j\leq r-1$, the increment in the rank of $\setS$ by collecting the coordinates of locality $j$ is $\kj{j}$. The increment in the rank of $\setS$ by locality $r$ coordinates is $\kj{r}-1$. (Note that this holds for Case (b) as well.) Moreover, for each $1\leq j \leq r$, when the algorithm is collecting the coordinates of locality $j$, the rank can increase by at most $j$ in each step (see~\eqref{eq:ti-bound}). Therefore, $\lj{j} \geq \ceillr{\frac{\kj{r}-1}{r}}$ for $1\leq j\leq r-1$ and $l_r \geq \ceillr{\frac{\kj{r}-1}{r}}$.

Combining this with $l = \sum_{j=1}^r\lj{j}$ gives,
\begin{equation}
\label{eq:l-LB-case-1}
l \geq \sum_{j=1}^{r-1}\ceillr{\frac{\kj{j}}{j}} + \ceillr{\frac{\kj{r}-1}{r}}.
\end{equation}

Substituting this into~\eqref{eq:S-case-1}, we get
\begin{IEEEeqnarray}{rCl}
|\setS| & \geq & k - 1 + \sum_{j=1}^{r-1}\ceillr{\frac{\kj{j}}{j}} + \ceillr{\frac{\kj{r}-1}{r}}\\
\label{eq:S-case-1-final}
& \geq & k - 2 + \sum_{j=1}^{r}\ceillr{\frac{\kj{j}}{j}}.
\end{IEEEeqnarray}

{\bf Case 2:} In the last step, we get $\rnk{\setSi{l-1} \cup \cwset{\gam{l}}} = k$.
For $1\leq i\leq l-1$, in the $i$-th iteration, we add $\cwset{\gam{i}}$. Thus, $\si \leq \loc{\cwi} + 1$. Further, vectors in $\cwset{\gam{i}}\setminus\setSi{i-1}$ are such that they yield a (possibly zero) vector in $\subspace{\setSi{i-1}}$. Therefore, for $1\leq i\leq l-1$, we get $\ti \leq \si - 1 \leq \loc{\cwi}$. 
In the last step $l$, we add $\cwset{\gam{l}'} \subset \cwset{\gam{l}}$. This increments $\rnk{\setS}$ by $t_l \geq 1$ (since $\rnk{\setSi{l-1}}\leq k - 2$), and $|\setS|$ by $s_l \geq t_l$. Therefore, we have
\begin{equation}
\label{eq:S-case-2}
|\setS| = \sum_{i = 1}^{l}\si \geq \sum_{i = 1}^{l-1}(\ti +1) + t_l = k - 1 + l - 1,
\end{equation}
the last equality follows from~\eqref{eq:S-l-total-rank}.  

{\it Lower bounding the number of iterations.} Similar to Case 1, in each iteration $i$ (including the last one), we have $\ti \leq \loc{\cwi}$. The only difference from Case 1 is that $\setS$ accumulates total rank of $k$ instead of $k-1$. Therefore, to lower bound $l$, we can use the same arguments as in Case 1 along with $\rnk{\setSi{l}}=k$ to obtain $l \geq \sum_{j=1}^{r-1}\ceillr{\frac{\kj{j}}{j}} + \ceillr{\frac{\kj{r}}{r}}$ in place of~\eqref{eq:l-LB-case-1}. Substituting this into~\eqref{eq:S-case-2} yields $|\setS| \geq k - 2 + \sum_{j=1}^{r}\ceillr{\frac{\kj{j}}{j}}$ (which is same as~\eqref{eq:S-case-1-final}).

Finally, noting that  $|\setS| \leq n - d$ from Fact~\ref{fact} and using this lower bound on $|\setS|$ gives~\eqref{eq:dist-bound-info-loc}.
\end{proof}

\subsection{Code Construction: Pyramid Codes}
\label{sec:pyramid-codes}
We show that the {\it parity splitting} construction of the Pyramid codes~\cite{Huang:07} can be adapted to obtain codes with unequal information locality, that are optimal with respect to~\eqref{eq:dist-bound-info-loc}. Consider an information locality profile $\kprof = \{\kj{1},\ldots,\kj{r}\}$. Let $\{\jp{1},\ldots,\jp{m}\}$ with $\jp{1} < \cdots < \jp{m}$ be the $m (\leq r)$ localities such that $\kj{\jp{p}} > 0$. We begin with a $(k+d-1,k,d)$ systematic maximum distance separable (MDS) code $\codep$. Let the representing coordinates be $\codeptsp = \{\ej{1},\ldots,\ej{k},\pj{0},\ldots,\pj{d-2}\}$, where $\ej{j}$ is the $j$-th column of a $k\times k$ identity matrix, and $\pj{j}$ for $0\leq j\leq d-2$ are the columns representing the parity coordinates. 

We partition the set $[k]$ into $m$ disjoint subsets $\setSi{1}, \ldots, \setSi{m}$ such that $|\setSi{p}| = \kj{\jp{p}}$ for each $p\in[m]$. Next, partition each subset $\setSi{p}$ into $\lj{p} = \ceillr{\frac{\kj{\jp{p}}}{\jp{p}}}$ disjoint subsets each of size at most $\jp{p}$. That is, $\setSi{p} = \cup_{i=1}^{\lj{p}}\setSi{p,i}$. For a vector $\vect{x}$ of dimension $k$, and a set $\setS\subseteq[k]$, let $\xS{\vect{x}}{\setS}$ denote the $|\setS|$-dimensional restriction of $\vect{x}$ to the coordinates  in set $\setS$. Then, we define the systematic code $\code$ with the following representation.
\begin{equation}
\label{eq:pyramid-code}
\codepts = \left\{\ej{1}, \ldots, \ej{k}, \xS{\pj{0}}{\setSi{1,1}}, \ldots, \xS{\pj{0}}{\setSi{1,\lj{1}}}, \xS{\pj{0}}{\setSi{2,1}}, \ldots, \xS{\pj{0}}{\setSi{2,\lj{2}}}, \ldots, \xS{\pj{0}}{\setSi{m,1}}, \ldots, \xS{\pj{0}}{\setSi{m,\lj{m}}},\pj{1},\ldots,\pj{d-2}\right\}.
\end{equation}

Note that we have {\it split} the parity $\pj{0}$ into $\sum_{j=1}^{r}\ceillr{\frac{\kj{j}}{j}}$ parities. Therefore, $n = k+d-2+\sum_{j=1}^{r}\ceillr{\frac{\kj{j}}{j}}$. It is easy to verify that parity splitting does not affect the distance, and hence, the code $\code$ has distance $d$. Since $\codep$ is an MDS code, we have $\wt{\pj{0}} = k$. Therefore, a set of $\kj{\jp{p}}$ information coordinates and $\ceillr{\frac{\kj{\jp{p}}}{\jp{p}}}$ parity coordinates  have locality at most $\jp{p}$ for each $p\in[m]$. Similar to the classical Pyramid codes in~\cite{Huang:07}, the last $d-2$ parity symbols may have locality as large as $k$. 

\section{Codes with Unequal All Symbol Locality}
\label{sec:all-symbol}
In this section, we extend the notion of information locality to profile to accommodate the codes whose parity symbols also have locality constraints. In this case, code symbols can be partitioned into disjoint subsets according to their locality, with maximum locality $\ra < k$. We define {\it all symbol locality profile} of a code as follows.

\begin{definition}
\label{def:loc-prof}
[All Symbol Locality Profile] Given an $\nkdq$ code $\code$, the all symbol locality profile of $\code$ is defined as a length-$n$ vector $\rprof(\code) = \{r_1,\ldots,r_n\}$, where $r_i$ is the locality of the $i$-th  coordinate of $\code$. Note that $1\leq \ri \leq k$ for each $i\in[n]$, assuming $d\geq 2$. 

Alternatively, we can specify the locality profile of $\code$ as a length-$\ra$ vector $\nprof(\code) = \{\nj{1}, \ldots, \nj{\ra}\}$, where $\ra = \max\{r_1, \ldots, r_n\}$ and $\nj{j}$ is the number of information coordinates of locality $j$ for $j\in[\ra]$. Note that $\forall j\in[\ra]$, $0\leq \nj{j}\leq n$, $\nj{\ra} \geq 1$ and $\sum_{j=1}^{\ra} \nj{j} = n$. 
\end{definition}

\begin{remark}
\label{rem:loc-prof}
For a code $\code$ with representation $\codepts$, let $\codeptsj{j}\subset\codepts$ be the subset of coordinates having locality $j$ for $1\leq j\leq \ra$. If $\nj{j} = 0$ for some $j$, then we set $\codeptsj{j} = \emptyset$. For $1\leq j\leq \ra$, we define 
\begin{equation}
\label{eq:k-j}
\kj{j} =  \rnk{\cup_{i=0}^{j}\codeptsj{i}} - \rnk{\cup_{i=0}^{j-1}\codeptsj{i}},
\end{equation}
where we set $\codeptsj{0} = \emptyset$. Define $r = \max\{j : \kj{j} > 0\}$. Then, $\{\kj{1}, \ldots, \kj{r}\}$ can be considered as the information locality profile of $\code$. Codes with the same all symbol locality profile can have different information locality profiles. 
\end{remark}

\subsection{Bound on the Minimum Distance}
\label{sec:all-sym- dist-bound}

Note that codes with unequal localities for all symbols are a special class of codes with unequal information localities. Therefore, the minimum distance upper bound in~\eqref{eq:dist-bound-info-loc} holds for an all symbol locality code having information locality profile $\kprof$. As noted in Remark~\ref{rem:loc-prof}, it is possible for a code to have different information locality profiles for a given all symbol locality profile. The upper bound in~\eqref{eq:dist-bound-info-loc} obtained using only information locality profile may not be tight for certain information localities. Our goal is to compute an upper bound on the minimum distance as a function of all symbol locality profile. 

\begin{theorem}
\label{thm:all-symbol-loc-bound}
Consider a code $\code$ with all symbol locality profile $\nprof = \{\nj{1},\ldots,\nj{\ra}\}$. Define $\kpj{j} = \nj{j} - \ceillr{\frac{\nj{j}}{j+1}}$. Let $r' = \max\{1\leq i \leq \ra : \sum_{j=1}^{i}\kpj{j} < k\}$. Let $r = \min\{r'+1\leq j\leq \ra : \nj{j}\geq 2\}$. Then, we have
\begin{equation}
\label{eq:all-symbol-loc-bound}
d \leq n - k + 2 - \sum_{j=1}^{r-1}\ceillr{\frac{\nj{j}}{j+1}} - \ceillr{\frac{k - \sum_{i=1}^{r-1}\left(\nj{j} - \ceillr{\frac{\nj{j}}{j+1}}\right)}{r}}.
\end{equation}
\end{theorem}
\begin{proof}
Similar to information locality case, we consider Algorithm~\ref{alg:find-S} to find a set $\setS\subset\codepts$ such that $\rnk{\setS}\leq k-1$. 

Recall that $\codeptsj{j}\subset\codepts$ is a subset of coordinates of locality $j$. Let $\kj{j} = \rnk{\cup_{i=0}^{j}\codeptsj{i}} - \rnk{\cup_{i=0}^{j-1}\codeptsj{i}}$, where we define $\codeptsj{0} = \emptyset$. 

It is easy to show that $\kj{j} \leq \kpj{j}$ for each $1\leq j\leq r$. In particular, consider the following greedy algorithm. Beginning with $\setTi{0}=\emptyset$ until $\setTi{p} = \codepts{j}$, in each iteration $p$, extend $\setTi{p-1}$ as by adding a coordinate $\vect{c}_p\in\codeptsj{j}\setminus\setTi{p-1}$ and all its repair group coordinates $\cwset{\rep{}{p}}$ to $\setTi{p-1}$. Specifically, $\setTi{p} = \setTi{p-1}\cup(\cwset{\gam{p}}\setminus\setTi{p-1})$. Now, in each iteration there must be at least one linear dependency between $\setTi{p-1}$ and $\cwset{\gam{p}}\setminus\setTi{p-1}$. Further, in each iteration, we extend the size of $\setT$ by at most $j$,  and thus, the number of iterations are at least $\ceillr{\frac{\nj{j}}{j+1}}$. Therefore, the number of linear dependencies among the coordinates in $\codeptsj{j}$ must be at least $\ceillr{\frac{\nj{j}}{j+1}}$.

{\it Case 1:} Suppose we have $\rnk{\setSi{i-1} \cup \cwset{\gam{i}}} \leq k-1$ throughout. Let $m$ be the locality of the last symbol picked by the algorithm. For $1\leq j\leq m-1$, the algorithm collects all the coordinates of locality $j$. Let $\nhj{m}\leq\nj{m}$ be the number of coordinates of locality $m$ that are collected by the algorithm. Then, we have
$$ |\setS| = \nj{1} + \cdots + \nj{m-1} + \nhj{m}.$$ 
Note that $\rnk{\setS}$ when $\setS$ has accumulated all the coordinates of locality up to $m-1$ is $\rnk{\cup_{j=1}^{m-1}\codeptsj{j}} = \sum_{j=1}^{m-1}\kj{j}$. Therefore, the rank accumulated from locality $m$ coordinates is $(k-1)-\sum_{j=1}^{m-1}\kj{j} :=\khj{m}$. Now, using standard arguments similar to the proof of Theorem~\ref{thm:dist-bound-info-loc}, it is easy to show that $\nhj{m} \geq \khj{m} + \ceillr{\frac{\khj{m}}{m}}$. Therefore,
\begin{equation}
\label{eq:S-LB-1}
|\setS| \geq \sum_{j=1}^{m-1} \nj{j}+ \khj{m} + \ceillr{\frac{\khj{m}}{m}} := |\setS|_{LB}.
\end{equation}

Next, we show that $|\setS|_{LB}$ is minimized when $\kj{j} = \kpj{j}$. Let $\setSp$ be the set collected if $\rnk{\cup_{i=0}^{j}\codeptsj{i}} - \rnk{\cup_{i=0}^{j-1}\codeptsj{i}} = \kpj{j}$. In this case the locality of the last coordinate must be $r$ provided $\sum_{j=1}^{r-1}\kpj{j} < k-1$. Let $\nhpj{r}$ be the number of coordinates of locality $r$ that are collected by the algorithm. (If $\sum_{j=1}^{r-1}\kpj{j} = k-1$, then $\nhpj{r}=0$ and the following analysis still holds.) Then, we have
$$|\setSp| = \nj{1} + \ldots + \nj{r-1} + \nhj{r}.$$
The rank accumulated in locality $r$ coordinates is $(k-1)-\sum_{j=1}^{r-1}\kpj{j} :=\khpj{r}$. Again, using standard arguments similar to the proof of Theorem~\ref{thm:dist-bound-info-loc}, it is easy to show that $\nhpj{r} \geq \khpj{r} + \ceillr{\frac{\khpj{r}}{r}}$. Therefore,
\begin{equation}
\label{eq:S-LB-2}
|\setSp| \geq \sum_{j=1}^{r-1} \nj{j}+ \khpj{r} + \ceillr{\frac{\khpj{r}}{r}} := |\setSp|_{LB}.
\end{equation}

Next, we show that $|\setSp|_{LB} \leq |\setS|_{LB}$. Suppose, for contradiction, $|\setSp|_{LB} > |\setS|_{LB}$. First, note that since $\kpj{j} \geq kj{j}$ for $1\leq j\leq r$, we have $r\leq m$. 

Case (1a): $m = r$. Then, we have
$$\sum_{j=1}^{r-1} \nj{j}+ \khpj{r} + \ceillr{\frac{\khpj{r}}{r}} > \sum_{j=1}^{r-1} \nj{j}+ \khj{r} + \ceillr{\frac{\khj{r}}{r}}.$$
However, this essentially implies $\sum_{j=1}^{r-1}\kpj{j} < \sum_{j=1}^{r-1}\kj{j}$, which is a contradiction.

Case (1b): $m < r$. Then, we have
$$\sum_{j=1}^{r-1} \nj{j}+ \khpj{r} + \ceillr{\frac{\khpj{r}}{r}} > \sum_{j=1}^{r-1} \nj{j}+ \nj{r} + \cdots + \khj{r} + \ceillr{\frac{\khj{r}}{r}}.$$
However, this implies $\khpj{r} + \ceillr{\frac{\khpj{r}}{r}} > \nj{r} + \cdots + \khj{r} + \ceillr{\frac{\khj{r}}{r}}$, which is a contradiction as $\khpj{r} + \ceillr{\frac{\khpj{r}}{r}} \leq \nhpj{r} \leq \nj{r}$.

Hence, to get smallest lower bound on $|\setS|$, one can assign maximum incremental rank $\kpj{j}$ to each locality $j$. Let $\lj{j}$ be the number of iterations during which Algorithm~\ref{alg:find-S} collects coordinates of locality $j$. Then, using the same arguments as in the proof of Theorem~\ref{thm:dist-bound-info-loc}, we have $|\setS| \geq k - 1 + \sum_{j=1}^{r}\lj{j}$ (see~\eqref{eq:S-case-1}). For $1\leq j\leq r-1$, the algorithm collects all the $\nj{j}$ coordinates of locality $j$. When a coordinate of locality $j$ is picked, the size of $\setS$ can be increased by at most $j+1$ in that iteration. Thus, $\lj{j} \geq \ceillr{\frac{\nj{j}}{j+1}}$ for $1\leq j \leq r-1$. For locality $r$, we increment the rank of $\setS$ by $(k-1)-\sum_{j=1}^{r-1}\kpj{j}$. At each step, tank is increased by at most $r$, thus $\lj{r} \geq \ceillr{\frac{(k-1)-\sum_{j=1}^{r-1}\kpj{j}}{r}}$. Hence,
$$|\setS| \geq k - 1 + \sum_{j=1}^{r-1} \ceillr{\frac{\nj{j}}{j+1}} + \ceillr{\frac{(k-1)-\sum_{j=1}^{r-1}\kpj{j}}{r}} \geq k - 2 + \sum_{j=1}^{r-1} \ceillr{\frac{\nj{j}}{j+1}} + \ceillr{\frac{k-\sum_{j=1}^{r-1}\kpj{j}}{r}}.$$

{\it Case 2:} In the last step, we get $\rnk{\setSi{l-1} \cup \cwset{\gam{l}}} = k$. Analysis to show that the smallest lower bound on $|\setS|$ is obtained assigning maximum incremental rank $\kpj{j}$ to each locality $j$is similar to Case 1. 

Using the same arguments as in the proof of Theorem~\ref{thm:dist-bound-info-loc}, we have $|\setS| \geq k - 2 + \sum_{j=1}^{r}\lj{j}$ (see~\eqref{eq:S-case-2}). Following the same argument as Case 1, $\lj{j} \geq \ceillr{\frac{\nj{j}}{j+1}}$ for $1\leq j \leq r-1$. For locality $r$, we increment the rank of $\setS$ by $k-\sum_{j=1}^{r-1}\kpj{j}$. At each step, tank is increased by at most $r$, thus $\lj{r} \geq \ceillr{\frac{k-\sum_{j=1}^{r-1}\kpj{j}}{r}}$. Hence,
$$|\setS| \geq k - 2 + \sum_{j=1}^{r-1} \ceillr{\frac{\nj{j}}{j+1}} + \ceillr{\frac{k-\sum_{j=1}^{r-1}\kpj{j}}{r}}.$$
Finally, the result follows from using the Fact~\ref{fact}.
\end{proof}

\subsection{Code Construction}
\label{sec:codes-all-symbol}
We adapt the rank-metric codes based LRC construction in~\cite{Silberstein:13} for the unequal all symbol locality scenario. The idea is to first precode the information symbols with a rank-metric code (in particular, with Gabidulin codes), and then use maximum distance separable (MDS) codes to obtain local parities. We begin with a brief review of rank-metric codes.

\subsection{Rank-Metric Codes}
\label{sec:rank-metric}
Let $\GFnm{q}$ be the set of all $N\times m$ matrices over $\GF{q}$. The {\it rank distance} is a distance measure between elements $A$ and $B$ of $\GFnm{q}$ defined as $\dr{A}{B} = \rnk{A-B}$. It can be shown that the rank distance is indeed a metric~\cite{Gabidulin:85}. A rank-metric code is a non-empty subset of $\GFnm{q}$ under the context of the rank metric. 

Typically, the rank-metric codes are considered by leveraging the correspondence between $\GF{q}^{1\times m}$ and an extension field $\GFm{q}$. By fixing a basis for $\GFm{q}$ as an $m$-dimensional vector space over $\GF{q}$, any element of $\GFm{q}$ can be represented as an $m$-length vector over $\GF{q}$. Similarly, any $N$-length vector over $\GFm{q}$ can be represented as an $N\times m$ matrix over $\GF{q}$. The rank of a vector $A\in\GFm{q}^N$ is the rank of $A$ as an $N\times m$ matrix over $\GF{q}$, which also works for the rank distance. This correspondence allows us to view a rank-metric code in $\GFnm{q}$ as a block code of length $N$ over $\GFm{q}$. 

Focussing on linear codes, an $(N,K,D)$ rank-metric code $\code\subseteq\GFm{q}^N$ is a linear block code over $\GFm{q}$ of length $N$, dimension $K$, and minimum rank distance $D$. For such codes, the Singleton bound becomes $d\leq \min\left\{1,\frac{m}{N}\right\}(N-K)+1$ (see~\cite{Gabidulin:85}). Codes that achieve this bound are called as maximum-rank distance (MRD) codes. Note that, for $m\geq N$, the Singleton bound for rank metric coincides with the classical Singleton bound for the Hamming metric. Indeed, when $m\geq N$, every MRD code is also MDS,  and hence can correct any $d-1$ {\it rank erasures}. 

{\bf Gabidulin Codes:} For $N\geq m$, a class of MRD codes was presented in~\cite{Gabidulin:85} by Gabidulin (see also~\cite{Delsarte:78}). 
A Gabidulin code can be obtained by evaluation of {\it linearized polynomials} defined as follows. A linearized polynomial $f(x)$ over $\GFm{q}$ of $q$-degree $K$ has the form $f(x) = \sum_{i=0}^{K} a_i x^{q^i}$, where $a_i\in\GFm{q}$ such that $a_K \neq 0$. Evaluation of a linearized polynomial is an $\GF{q}$-linear transform from $\GFm{q}$ to itself. In other words, for any $a,b \in \GF{q}$ and $x,y\in\GFm{q}$, we have $f(ax+by) = af(x)+bf(y)$. 

A codeword in an $(N,K,N-K+1)$ Gabidulin code $\code_{Gab}$ over $\GFm{q}$ for $m\geq N$ is defined as $\cw = \left(f(g_1),\ldots,f(g_N)\right) \in \GFm{q}^N$, where $f(x)$ is a linearized polynomial over $GFm{q}$ of $q$-degree $K-1$ whose coefficients are information symbols, and evaluation points $g_1,\ldots,g_N\in\GFm{q}$ are linearly  independent over $\GF{q}$. Note that since Gabidulin code is also an MDS code, it can correct any $N-K$ erasures.

\subsection{Code Construction}
\label{sec:construction}
In the following, we give a construction of an $\nkd$ LRC with all symbol locality profile $\nprof = \{\nj{1},\ldots,\nj{\ra}\}$ which attains the distance bound in~\eqref{eq:all-symbol-loc-bound}. For the simplicity of presentation, we assume that $j+1\mid\nj{j}$ for each $j$. One can generalize the construction for the case when this is not the case.

{\bf Construction 1.} Consider a length-$k$ vector of information symbols $\m \in \GFm{q}^k$.  First, we precode $\m$ using a Gabidulin code. Then, the codeword of the Gabidulin code is partitioned into local groups, and the local parities are computed for each group using MDS codes over $\GF{q}$. The details are as follows.

Define $\Nj{j} = \nj{j}\left(\frac{j}{j+1}\right)$ for each $j\in[\ra]$. Let $N = \sum_{j=1}^{\ra}\Nj{j}$. Encode $\m$ using an $(N,k,N-k+1)$ Gabidulin code to obtain $\cwG\in\GFm{q}^N$. Partition $\cwG$ into $\ra$ disjoint groups $\cwG = \cup_{j=1}^{\ra}\cwGj{j}$ such that $|\cwGj{j}| = \Nj{j}$ for $j\in[\ra]$ with $\cwGj{j} = \emptyset$ for each $j$ such that $\Nj{j} = 0$. For each $1\leq j\leq \ra$ such that $\Nj{j} > 0$, further partition $\cwGj{j}$ symbols into $\frac{\Nj{j}}{j}$ disjoint local groups each of size $j$, \ie, $\cwG{j} = \cup_{i=1}{\frac{\Nj{j}}{j}}\cwGj{j,i}$. For each group $\cwGj{j,i}$ of $j$ symbols, generate a local parity using a $(j+1,j,2)$ MDS code over $\GF{q}$. Denote the resulting code as $\code_{\textrm{LRC}}$. Note that the total number of symbols are $\sum_{j=1}^{\ra}\frac{\Nj{j}}{j}(j+1) = \sum_{j=1}^{\ra}\nj{j} = n$. Note that, we generate the local parities in such a way that $\code_{\textrm{LRC}}$ possesses all symbol locality profile $\{\nj{1},\ldots,\nj{\ra}\}$.

Next, we show that the above construction achieves the distance bound mentioned in Theorem~\ref{thm:all-symbol-loc-bound}.
\begin{theorem}
\label{thm:gabidulin-codes-based-LRC}
Let $\code_{\textrm{LRC}}$ be an $(n,k,d)$ LRC with all symbol locality profile $\{\nj{1},\ldots,\nj{\ra}\}$ obtained by Construction 1. If $j+1\mid\nj{j}$ for each $j\in[\ra]$, then $\code_{\textrm{LRC}}$ over $\GFm{q}$ for $m\geq \sum_{j=1}^{\ra}\nj{j}\left(\frac{j}{j+1}\right)$ and $q\geq \ra + 1$, achieves the bound in~\eqref{eq:all-symbol-loc-bound}.
\end{theorem}
\begin{proof}
Similar to~\cite{Silberstein:13}, the idea is show that any $e : = n - k + 1 - \sum_{j=1}^{r-1}\ceillr{\frac{\nj{j}}{j+1}} - \ceillr{\frac{k - \sum_{i=1}^{r-1}\left(\nj{j} - \ceillr{\frac{\nj{j}}{j+1}}\right)}{r}}$ symbol erasures correspond to $N-K$ rank erasures, which can be corrected by the Gabidulin code. 

The $\GF{q}$-linearity of the linearized polynomials plays a crucial role. In particular, since the local parities are obtained using an MDS code over $\GF{q}$, any symbol $\cwi$ of locality $j$ can be written as $\cwi = \sum_{p=1}^{j} a_p\cw_{i_p} =  \sum_{p=1}^{j} a_p f(g_{i_p}) = f\left(\sum_{p=1}^{j} a_p g_{i_p}\right)$. Hence, for each $j\in[\ra]$, in a local group of size $j$, any $m\leq j$ symbols are evaluations of $f(x)$ in $m$ points that are linearly independent over $\GF{q}$. Therefore, for each $j\in[\ra]$, in a local group of size $j+1$, any $i+1 (\leq j+1)$ symbol erasures correspond to $i$ rank erasures. Moreover, taking any $j$ points from all local groups of size $j+1$ for each $j\in[\ra]$, we obtain the Gabidulin codeword, which has obtained by precoding $\m$.

With above observation, the worst case erasure pattern is when the erasures occur in the smallest possible number of local groups (of possibly different localities), and the number of erasures in each local group are maximal.

Note that we can write $n$ as $n = \sum_{j=1}^{\ra}\Nj{j} + \frac{\Nj{j}}{j}$. Let $k = \sum_{j=1}^{r-1} \Nj{j} + \Npj{r}$ for some $\Npj{r} < \Nj{r}$. Then, we can write 
\begin{equation}
\label{eq:erasures}
e = 1 + \sum_{j=r}^{\ra}\left(\Nj{j} + \frac{\Nj{j}}{j}\right) -  \left(\Npj{r}  + \ceillr{\frac{\Npj{r}}{r}}\right).
\end{equation} 
%Before simplification, $e = 1 + \sum_{j=1}^{\ra}(\Nj{j} + \frac{\Nj{j}}{j}) - (\sum_{j=1}^{r-1} \Nj{j} + \Npj{r}) - \sum_{j=1}^{r-1}\frac{\nj{j}}{j+1} - \ceillr{\frac{k - \sum_{i=1}^{r-1}\left(\nj{j} - \frac{\nj{j}}{j+1}\right)}{r}}$. 
On the other hand, for the outer Gabidulin code, we have 
\begin{equation}
\label{eq:rank-distance}
N-k = \sum_{j=r}^{\ra}\Nj{j} - \Npj{r}.
\end{equation} 
% Before simplification, $N-k = \sum_{j=1}^{\ra}\Nj{j} - (\sum_{j=1}^{r-1} \Nj{j} + \Npj{r})$.

{\it Case 1:} $r\mid\Npj{r}$. Let $\Npj{r} = r\beta$. Then, from~\eqref{eq:erasures}, we have $e = 1 + \sum_{j=r+1}^{\ra}(j+1)\left(\frac{\nj{j}}{j+1}\right) + (r+1)\left(\frac{\nj{r}}{r+1} - \beta\right)$. Thus, in the worst case, the number of local groups that are completely erased are $\sum_{j=r+1}^{\ra}\left(\frac{\nj{j}}{j+1}\right) + \left(\frac{\nj{r}}{r+1} - \beta\right)$ with one erasure in an additional group. Recall that, due to the $\GF{q}$-linearity, any $i+1$ erasures in a local group of size $j+1$, the number of rank erasures corresponding to the Gabidulin codeword are only $j$. Thus, total number of rank erasures are $\sum_{j=r+1}^{\ra}j\left(\frac{\nj{j}}{j+1}\right) + r\left(\frac{\nj{r}}{r+1} - \beta\right)$.

However, from~\eqref{eq:rank-distance}, we get $N-k = \sum_{j=r+1}^{\ra}j\left(\frac{\nj{j}}{j+1}\right) + r\left(\frac{\nj{r}}{r+1} - \beta\right)$. Therefore, all the rank erasures can be corrected by the outer Gabidulin code.

{\it Case 2:} $r\nmid\Npj{r}$. Let $\Npj{r} = r\beta + \gamma$, where $1\leq\gamma\leq r-1$. Then, from~\eqref{eq:erasures}, we have $e = 1 + \sum_{j=r+1}^{\ra}(j+1)\left(\frac{\nj{j}}{j+1}\right) + (r+1)\left(\frac{\nj{r}}{r+1} - \beta - 1\right) + (r-\gamma+1)$. In other words, in the in the worst case, the number of local groups that are completely erased are $\sum_{j=r+1}^{\ra}\left(\frac{\nj{j}}{j+1}\right) + \left(\frac{\nj{r}}{r+1} - \beta - 1\right)$ with $(r-\gamma+1)$ erasures in an additional group. This corresponds to $\sum_{j=r+1}^{\ra}j\left(\frac{\nj{j}}{j+1}\right) + r\left(\frac{\nj{r}}{r+1} - \beta - 1\right) + (r - \gamma)$ rank erasures. 

From~\eqref{eq:rank-distance}, we get $N-k = \sum_{j=r+1}^{\ra}j\left(\frac{\nj{j}}{j+1}\right) + r\left(\frac{\nj{r}}{r+1} - \beta - 1\right) + (r - \gamma)$. Hence, all the rank erasures can be corrected by the outer Gabidulin code.
\end{proof}

\section{Information Locality Requirement}
\label{sec:loc-req}

In general, one can design codes for different locality profiles, which gives rise to the following natural question: how to choose a locality profile that gives largest minimum distance. Towards this, we define a notion of {\it locality requirement} as follows.

\begin{definition}
\label{def:loc-req}
Let $\kprofreq = \{\ktj{1}, \ldots, \ktj{r}\}$ be a length-$r$ vector for some $r < k$ such that for each $1\leq j\leq r$, we have $0\leq \ktj{j}\leq k$ and $\sum_{j = 1}^r \ktj{j} = k$. Consider a code $\code$ with information locality profile $\kprof = \{\kj{1}, \ldots, \kj{r'}\}$ for some $r' \leq r$. We say that $\code$ satisfies information locality requirement $\kprofreq$ if, for each $1\leq i\leq r$, we have $\sum_{j=1}^{i}\kj{j} \geq \sum_{j=1}^{i}\ktj{j}$, where we set $\kj{j} = 0$ for $r'+1\leq j \leq r$ if $r' < r$. Further, in this case, we say that locality profile $\kprof$ respects locality requirement $\kprofreq$, and denote this as $\kprof \succeq \kprofreq$.
\end{definition}

%\begin{remark}
%\label{rem:loc-req}
Different locality profiles can respect a locality requirement $\kprofreq$, and one can ask which locality profile would give larger minimum distance. For example, let $\kprofreq = \{0, 3, 3\}$. Then, one can find a number of locality profiles that respect $\kprofreq$, such as $\kprof_1 = \{2,4,0\}$, $\kprof_2 = \{3,0,3\}$, $\kprof_3 = \{0, 6, 0\}$, $\kprof_4 = \{1,2,3\}$. Among these, the last two locality profiles would give the largest minimum distance.  However, in general, since a large number of locality profiles can respect a locality requirement, it is not clear how to find an optimal locality profile with respect to minimum distance. %In the following we present a greedy algorithm to find an optimum locality profile that gives the largest bound on the minimum distance.
%\end{remark}
 
Give a locality requirement $\kprofreq$, we are interested in finding a locality profile $\kprof \succeq \kprofreq$ which results in largest upper bound on the minimum distance for fixed $n$. More formally, we can define the problem as follows.

\begin{IEEEeqnarray}{lCr}
%{\textrm{Find}} & \kprof^* = \arg\min_{\kprof\in\mathbb{Z}_{+}^{r}} \sum_{j=1}^r \ceillr{\frac{\kj{j}}{j}}\\
\label{eq:sum-of-ceils}
\min_{\kprof\in\mathbb{Z}_{+}^{r}} & \sum_{j=1}^r \ceillr{\frac{\kj{j}}{j}} & {(P1)}\\
{\textrm{s.t.}} & \sum_{j=1}^{i}\kj{j} \geq \sum_{j=1}^{i}\ktj{j}, & \:\:{\textrm{for}}\:\: 1\leq i\leq r,\\
{\textrm{and}} & \sum_{j=1}^{r} \kj{j} = \sum_{j=1}^{r}\ktj{j} & {}.
\end{IEEEeqnarray} 
 
A solution of the above optimization problem is said to be an {\it optimal} locality profile. In the following we give a greedy algorithm which finds an optimal $\kprofopt$. From $\sum_{j=1}^{r-1}\kj{j}\geq\sum_{j=1}^{r-1}\ktj{j}$ and $\sum_{j=1}^{r}\kj{j} = \sum_{j=1}^{r}\ktj{j}$, we get that $\kj{r}\leq\ktj{r}$. In similar way, we can see that the inequality constraints above can be replaced by $\sum_{j=i}^{r}\kj{j} \geq \sum_{j=i}^{r}\ktj{j}$ and $\sum_{j=1}^{r}\kj{j} = \sum_{j=1}^{r}\ktj{j}$. The idea of the algorithm is to start with the largest locality $r$ and set $\ksj{r}$ as the largest multiple of $r$ such that $\ksj{r} \leq \ktj{r}$. Move the residue $\ktj{r} - \ksj{r}$ to the next locality $r-1$, and set $\ksj{r-1}$ as the largest multiple of $r$ such that $\ksj{r-1} \leq \ktj{r-1}+\ktj{r}-\ksj{r}$. We continue this until we reach locality 1.

%%%%%%%%%%%%%%%%%%%%%%%%%%%%%%%%%%%%%%%%%%%%%%%%%
\begin{algorithm}[!t]
\caption{Find an optimal locality profile $\kprofopt$ for a given locality requirement $\kprofreq$}
\label{alg:find-k-opt}
\begin{algorithmic}[1]
\STATE{Set $\gj{r+1} = 0$, $j = r$} 
\WHILE{$j\geq 1$}
	\STATE{Chose integers $\bj{j}$ and $\gj{j}$ such that $\ktj{j} + \gj{j+1} = j\bj{j} + \gj{j}$}
        \STATE{Set $\ksj{j} = j\bj{j}$}
	\STATE{Decrement $j$}	
\ENDWHILE
\end{algorithmic}
\end{algorithm}
%%%%%%%%%%%%%%%%%%%%%%%%%%%%%%%%%%%%%%%%%%%%%%%%%

\begin{remark}
\label{rem:greedy-ksj}
Note that Algorithm~\ref{alg:find-k-opt} assigns $\ksj{j} = \ktj{j} + \gj{j+1} - \gj{j}$ for each locality $j$. This gives $\sum_{j=i}^{r}\ksj{j} = \sum_{j=i}^{r}\ktj{j} - \gj{i}$ for each $r\geq i\geq 1$.
\end{remark}

\begin{theorem}
\label{thm:ks-is-opt}
Given an information locality requirement $\kprofreq$, the information locality profile $\kprofopt$ given by Algorithm~\ref{alg:find-k-opt} results in the largest upper bound on the minimum distance among all the information locality profiles that respect the given information locality requirement. 
\end{theorem} 
\begin{proof}
The idea is to show that any optimal information locality profile can be transformed into a form of $\kprofopt$ without loosing optimality. We first prove that it is always possible to obtain an optimal information locality profile $\kprofpp$ such that $j\mid\kppj{j}$ for each $j\in[r]$.

\begin{lemma}
\label{lem:k'-to-k''}
Given a locality requirement $\kprofreq$, any optimal information locality profile $\kprofp$ can be converted into another optimal information locality profile $\kprofpp$ such that $j \mid \kppj{j}$ $\forall j\in[r]$.
\end{lemma}
\begin{proof}
By induction on the number of localities $j$ such that $j \nmid \kpj{j}$. Let $|\{j : j\nmid\kpj{j}\}| = m$.

Basis step: $m = 1$. Let $\jp{m}\in[r]$ be the only locality such that $\jp{m}\nmid\kpj{\jp{m}}$. We can write $\kpj{\jp{m}} = \jp{m}\bj{\jp{m}} + \gj{\jp{m}}$ such that $1\leq\bj{\jp{m}}\leq\jp{m}-1$. Set $\kppj{\jp{m}} = \kpj{\jp{m}} - \gj{\jp{m}}$, $\kppj{\gj{\jp{m}}} = \kpj{\gj{\jp{m}}} + \gj{\jp{m}}$, and $\kppj{j} = \kpj{j}$ for all $j\in[r]$ such that $j\neq\jp{m}$, $j\neq\gj{\jp{m}}$.

First, observe that $\kprofpp$ is such that $j\mid\kj{j}$ for each $j\in[r]$, since $\gj{\jp{m}}\mid\kpj{\gj{\jp{m}}}$. 

Second, note that $\kprofpp$ is a feasible solution for (P1). This is because, for $1\leq i\leq \jp{m}-1$, we have $\sum_{j=1}^{i}\kppj{j}  = \sum_{j=1}^{i}\kpj{j} + \gj{\jp{m}} \geq \sum_{j=1}^{i}\ktj{j}$, and for $\jp{m}\leq i\leq r$, we have $\sum_{j=1}^{i}\kppj{j}  = \sum_{j=1}^{i}\kpj{j} \geq \sum_{j=1}^{i}\ktj{j}$. For both these cases, the inequality follows since $\kprofp$ satisfies the constraints of (P1).

Finally, it is easy to see that $\kprofpp$ is also optimal, since $\ceillr{\frac{\kppj{\jp{m}}}{\jp{m}}} = \ceillr{\frac{\kpj{\jp{m}}}{\jp{m}}} - 1$, $\ceillr{\frac{\kppj{\gj{\jp{m}}}}{\gj{\jp{m}}}} = \ceillr{\frac{\kpj{\gj{\jp{m}}}}{\gj{\jp{m}}}} + 1$, and $\ceillr{\frac{\kppj{j}}{j}} = \ceillr{\frac{\kpj{j}}{j}}$ for the rest of the localities.

Induction step: $m \geq 2$. Suppose the hypothesis holds whenever $|\{j : j\nmid\kpj{j}\}| \leq m-1$. Consider the case when $|\{j : j\nmid\kpj{j}\}| = m$. Denote such a set of localities as $\{\jp{1},\ldots,\jp{m}\}$, where $\jp{1} < \cdots < \jp{m}$. Now, we can write $\kpj{\jp{m}} = \jp{m}\bj{\jp{m}} + \gj{\jp{m}}$ such that $1\leq\bj{\jp{m}}\leq\jp{m}-1$. Set $\kpj{\jp{m}} = \kpj{\jp{m}} - \gj{\jp{m}}$, and $\kpj{\gj{\jp{m}}} = \kpj{\gj{\jp{m}}} + \gj{\jp{m}}$.

Similar to $m=1$ case, we can verify that $\kprofp$ remains to be an optimal solution to (P1) after the  transformation. Further, since $\jp{m}\mid\kpj{\jp{m}}$, we get $|\{j : j\nmid\kpj{j}\}| = m-1$. Then, the proof follows by the induction hypothesis.
\end{proof}

Let $|\{j : \kppj{j}\neq\ksj{j}\}| = m$. Denote such a set of localities as $\{\jp{1},\ldots,\jp{m}\}$, where $\jp{1} < \cdots < \jp{m}$. We first prove some properties for the localities where the coordinate values differ. %highest locality $\jp{m}$, where $\kprof''$ differs from $\kprofopt$.

\begin{proposition}
\label{prop:last-index}
$\kppj{\jp{m}} < \ksj{\jp{m}}$
\end{proposition}
\begin{proof}
Suppose, for contradiction, $\kppj{\jp{m}} > \ksj{\jp{m}}$. We can write $\kppj{\jp{m}} = \ksj{\jp{m}} + p\jp{m}$ for some integer $p\geq 1$, since both $\kppj{\jp{m}}$ and $\ksj{\jp{m}}$ are multiples of $\jp{m}$. Consider
\begin{IEEEeqnarray}{rCl}
\sum_{i=\jp{m}}^{r} \kppj{i} & = & \ksj{\jp{m}} + p\jp{m} + \sum_{i=\jp{m}+1}^{r} \ksj{i}\\
& = & \sum_{i=\jp{m}}^{r} \ktj{i} - \gj{\jp{m}} + p\jp{m}\\
& \geq & \sum_{i=\jp{m}}^{r} \ktj{i} - (\jp{m} - 1) + p\jp{m}\\ 
& \geq & \sum_{i=\jp{m}}^{r} \ktj{i} + (p - 1)\jp{m} + 1\\
& \geq & \sum_{i=\jp{m}}^{r} \ktj{i}.
\end{IEEEeqnarray}
However, this contradicts the feasibility of $\kprofpp$ as it should satisfy $\sum_{\jp{m}}^{r} \kppj{i}\leq \sum_{\jp{m}}^{r} \ktj{i}$ (due to $\sum_{i=1}^{\jp{m}} \kppj{i}\geq \sum_{i=1}^{\jp{m}} \ktj{i}$ and $\sum_{i=1}^{r} \kppj{i}  = \sum_{i=1}^{r} \ktj{i}$).
\end{proof}

Next, we show that for any information locality profile, moving the coordinates to the higher locality does not increase the minimum distance bound. %For simplicity, we focus on optimal information locality profile. 
\begin{proposition}
\label{prop:move-to-right}
Consider an information locality profile $\kprof$. For any locality pair $i$ and $j$ such that $i<j$ and $\kj{j}>0$. Set $\kj{i} = \kj{i} - \delta$ and $\kj{j} = \kj{j} + \delta$ for an integer $\delta$ such that either $i\mid\delta$ or $j\mid\delta$ (or both). Then, such a transformation does not increase the value of the minimum distance bound. %(Note that this may affect the feasibility of $\kprof$.) 
\end{proposition}
\begin{proof}
Case 1: $i\mid\delta$. Let $\delta = ia$ for some integer $a$. After moving the coordinates of locality $i$ to locality $j$, the term $\ceillr{\frac{\kj{i}}{i}}$ reduces by $a$. Whereas, the term $\ceillr{\frac{\kj{j}}{j}}$ increases by at most $\ceillr{\frac{\delta}{j}}$, which itself is at most $a$. 

Case 2: $j\mid\delta$. Let $\delta = jb$ for some integer $b$. In this case, the term $\ceillr{\frac{\kj{j}}{j}}$ increases by $b$. Whereas, the term $\ceillr{\frac{\kj{i}}{i}}$ reduces by at least $\floorlr{\frac{\delta}{i}}$, which itself is at least $b$.

Therefore, in both the above case, the value of~\eqref{eq:sum-of-ceils} does not increase.
\end{proof}

Finally, we show that for any information locality profile, moving the coordinates to the lower locality to obtain divisibility does not change the minimum distance bound.
\begin{proposition}
\label{prop:move-to-left}
Consider an information locality profile $\kprof$. Let $j$ be a locality such that $j\nmid\kj{j}$, and let $\kj{j} = j\bj{j} + \gj{j}$ for some integers $\bj{j}$ and $1\leq\gj{j}\leq j-1$. Then, setting $\kj{j} = \kj{i} - \gj{j}$ and $\kj{\gj{j}} = \kj{\gj{j}} + \gj{j}$ does not change the value of the minimum distance bound. 
\end{proposition}
\begin{proof}
The argument is the same as for the basis step in the proof of Lemma~\ref{lem:k'-to-k''}.
\end{proof}

Finally, we show that we can transform an optimal information locality profile where divisibility holds for each locality into $\kprofopt$. 
\begin{lemma}
\label{lem:k''-to-ks}
Given a locality requirement $\kprofreq$, any optimal information locality profile $\kprofp$, where $j\mid\kppj{j}$ for each $j$, can be converted into $\kprofopt$ without loosing optimality, where $\kprofopt$ is the output of Algorithm~\ref{alg:find-k-opt}.
\end{lemma}
\begin{proof}
We give an iterative algorithm (Algorithm~\ref{alg:k''-to-ks}) to transform an optimal information locality profile $\kprofpp$ to $\kprofopt$. First note that, by Proposition~\ref{prop:last-index}, it must be that $\kppj{\jp{m}} < \ksj{\jp{m}}$ in the first iteration of the outer while-loop. Moreover, at line~\ref{inner-loop}, $\kprofpp$ is such that $j\mid\kppj{j}$ for each $j\in[r]$, hence we can invoke Proposition~\ref{prop:last-index} for the every iteration of outer while-loop. Next, the optimality of $\kprofpp$ is maintained at line~\ref{move-right} due to Proposition~\ref{prop:move-to-right}, and also at line~\ref{move-left} due to Proposition~\ref{prop:move-to-left}. Finally, Algorithm~\ref{alg:k''-to-ks} must terminate in finite time as $m$ decreases by at least 1 at line~\ref{inner-loop}. 
\end{proof}

%%%%%%%%%%%%%%%%%%%%%%%%%%%%%%%%%%%%%%%%%%%%%%%%%
\begin{algorithm}[!t]
\caption{Transform an optimal locality profile $\kprofpp$ to $\kprofopt$}
\label{alg:k''-to-ks}
\begin{algorithmic}[1]
\STATE{Let $|\{j : \kppj{j}\neq\ksj{j}\}| = m$} 
\WHILE{$m > 0$}
	\STATE{Let $\jp{m} = \max\{j : \kppj{j}\neq\ksj{j}\}$}
	\WHILE{$\kppj{\jp{m}} < \ksj{\jp{m}}$}
		\STATE{Let $\jp{p} = \max\{j : \kppj{j} > \ksj{j}\}$}
		\STATE{Let $\delta_{\jp{m}} = \ksj{\jp{m}} - \kppj{\jp{m}}$, 
				   $\delta_{\jp{p}} = \kppj{\jp{p}} - \ksj{\jp{p}}$}
		\label{move-right}
		\STATE{Set $\kppj{\jp{m}} = \kppj{\jp{m}} + \min\{\delta_{\jp{m}}, \delta_{\jp{p}}\}$, 
		                    $\kppj{\jp{p}} = \kppj{\jp{p}} - \min\{\delta_{\jp{m}}, \delta_{\jp{p}}\}$}
        		\IF{$\delta_{\jp{m}} < \delta_{\jp{p}}$}
			\STATE{Let $\kppj{\jp{p}} = \jp{p}\bj{\jp{p}} + \gj{\jp{p}}$}
			\IF{$\gj{\jp{p}}>0$}
				\label{move-left}
				\STATE{Set $\kppj{\jp{p}} = \kppj{\jp{p}} - \gj{\jp{p}}$, 
				                    $\kppj{\gj{\jp{p}}} = \kppj{\gj{\jp{p}}} + \gj{\jp{p}}$}
			\ENDIF
		\ENDIF
	\label{inner-loop}	
	\ENDWHILE
	\STATE{Set $m = |\{j : \kppj{j}\neq\ksj{j}\}|$}
\ENDWHILE
\end{algorithmic}
\end{algorithm}
%%%%%%%%%%%%%%%%%%%%%%%%%%%%%%%%%%%%%%%%%%%%%%%%%

The proof of Theorem~\ref{thm:ks-is-opt} follows from Lemma~\ref{lem:k'-to-k''} and Lemma~\ref{lem:k''-to-ks}

\end{proof}

\section*{Acknowledgment}
Swanand Kadhe would like to thank Ankit Singh Rawat for helpful discussions and for pointing out LRC constructions based on rank-metric codes; and also thank Anoosheh Heidarzadeh for helpful discussions, especially on locality requirement.

%\nocite{*}
\bibliographystyle{IEEEtran}
\bibliography{Bib_unequal_LRC_v1,IEEEabrv}

\end{document}